\documentclass[sigplan,10pt,screen]{acmart}

\usepackage{xspace}
\usepackage{hyperref}
\usepackage{makecell}
\usepackage{mathpartir,logicproof,mathtools}

\newcommand{\N}{\mathbb{N}}

\newcommand{\U}{\mathsf{U}}
\newcommand{\ok}{\ \mathsf{ok}}
\newcommand{\lfl}{\textsf{Lean4Lean}\xspace}
\newcommand{\type}{\ \mathsf{type}}
\newcommand{\scott}[1]{\llbracket #1\rrbracket}

\newcommand*{\axiom}[2][]{\infer[#1]{}{#2}}
\newcommand{\code}[1]{\mbox{\ttfamily\small #1}}
\newcommand{\hole}{\code{\_}}

\DeclareMathOperator{\suc}{S}
\DeclareMathOperator{\imax}{imax}

\DeclareUnicodeCharacter{2203}{$\exists$}
\DeclareUnicodeCharacter{2227}{$\land$}

\usepackage{xcolor}
\definecolor{keywordcolor}{rgb}{0.7, 0.1, 0.1}   
\definecolor{tacticcolor}{rgb}{0.0, 0.1, 0.6}    
\definecolor{commentcolor}{rgb}{0.4, 0.4, 0.4}   
\definecolor{symbolcolor}{rgb}{0.0, 0.1, 0.6}    
\definecolor{sortcolor}{rgb}{0.1, 0.5, 0.1}      
\definecolor{attributecolor}{rgb}{0.7, 0.1, 0.1} 

\usepackage{listings}

\newcommand{\sref}[2]{\hyperref[#2]{#1~\ref{#2}}}

\begin{document}

\title{Lean4Lean: Verifying a Typechecker for Lean, in Lean}

\author{Mario Carneiro}
\email{marioc@chalmers.se}
\orcid{0000-0002-0470-5249}
\affiliation{%
  \institution{Chalmers University of Technology}
  \institution{University of Gothenburg}
  \city{Gothenberg}
  \country{Sweden}
}

\renewcommand{\shortauthors}{M. Carneiro}

\begin{abstract}
In this paper we present a new ``external checker'' for the Lean theorem prover, written in Lean itself. This is the first complete typechecker for Lean 4 other than the reference implementation in C++ used by Lean itself, and our new checker is competitive with the original, running between 20\% and 50\% slower and usable to verify all of Lean's mathlib library, forming an additional step in Lean's aim to self-host the full elaborator and compiler. Moreover, because the checker is written in a language which admits formal verification, it is possible to state and prove properties about the kernel itself, and we report on progress to formalize the Lean type theory abstractly and prove some theorems about it. Finally, we combine these to get a proof of correctness of parts of the kernel. We plan to use this project to help justify any future changes to the kernel and type theory and ensure unsoundness does not sneak in through either the abstract theory or implementation bugs. The verification is already paying off, as one soundness bug has been spotted and fixed as a result of this work.
\end{abstract}




\maketitle

\section{Introduction}

\textsc{Lean} \cite{demoura} is a theorem prover based on the Calculus of Inductive Constructions (\textsf{CIC}) \cite{cic}, quite similar to its older brother \textsc{Rocq}\footnote{Formerly known as Coq.} \cite{coq2020}, but as the name suggests, one of the ways it sought to differentiate itself was the desire to make the kernel ``leaner,'' relying on simpler primitives while retaining most of the power of the system. This was in particular informed by periodic news of unsoundnesses in Rocq\footnote{\url{https://github.com/rocq-prover/rocq/blob/master/dev/doc/critical-bugs.md}} due to complications of the many interacting features in the kernel. So from the beginning, Lean advertised its simpler foundations and multiple external checkers \cite{tc,trepplein} as reasons that it would not suffer from the same troubles. These hopes were broadly upheld: for the entire release history of Lean 3, there were no soundness bugs reported against the kernel (written in C++), and Carneiro \cite{leantt} showed the consistency of Lean with respect to ZFC with $n$ inaccessible cardinals for all $n<\omega$, so the soundness story seemed fairly strong.

But times move on, and Lean 4 now exists as an (almost) ground-up rewrite of Lean with most components now written in Lean itself. The kernel was one of the few components that was not rewritten, but it was extended with various features for performance reasons like bignum arithmetic, nested inductive types, and primitive projections, and unfortunately some soundness bugs crept in during the process and Lean's record is no longer spotless. Another thing that was lost during the port was the external checkers: the whole metaprogramming infrastructure was redesigned so old external checkers are no longer applicable to Lean 4. While most of the new features can be treated as mere abbreviations, nested inductive types and $\eta$ for structures significantly impact the theory, and as a result the soundness proof from \cite{leantt} is no longer directly applicable. (But see \cite{sebastianthesis}, which covers some of these modifications.)

We interpret this as a cautionary tale: It is not sufficient to see a prior system and hope to do better through sheer force of will. We are all human, and mistakes in both the theory and the implementation happen as a matter of course. The way to do better is not to rewrite, but to set up a process that structurally ensures that mistakes are either prevented or corrected. For a proof assistant, there are three main ways to do this: (1) \emph{testing}, (2) \emph{independent reimplementation}, and (3) \emph{formal verification}. Lean has a fairly robust test suite, but testing helps more in the high-frequency low-impact domain; soundness bugs are rare and generally occur in areas that are insufficiently tested because the developer did not even consider the error condition.

The main contribution of this paper lies in (2): an independent reimplementation. To restore the soundness story of Lean to a satisfactory state, we believe it is necessary to bring back the external checkers. And what better language to do so than Lean itself?  There are reasons not to use Lean for a truly \emph{independent} verifier (and more external checkers are coming, in Scala\footnote{\url{https://github.com/gebner/trepplein}} and Rust\footnote{\url{https://github.com/ammkrn/nanoda_lib}}), but there are also advantages to having a kernel written in Lean:
\begin{itemize}
  \item The Lean elaborator is already written in Lean as a form of ``whitebox automation'': users are able to make use of elaborator APIs to write their own tactics with just as much flexibility as the core language itself. Similarly, having a ``whitebox kernel'' makes it easier for Lean users to query the internal state and understand how Lean reduces or typechecks terms.
  \item Unlike Rust and Scala, Lean is designed for both programming and proving, so a kernel written in Lean can be formally verified. And we have good reasons for wanting the kernel to satisfy certain properties\ldots
\end{itemize}

In this paper, we will present \href{https://github.com/digama0/lean4lean}{\lfl},\footnote{\url{https://github.com/digama0/lean4lean}} an external verifier for Lean 4, written in Lean. The verifier itself is more than a prototype: it is fully capable of checking the entirety of \textsf{Lean} (the Lean compiler package), \textsf{Batteries} (the extended standard library), and \textsf{Mathlib} (the Lean mathematics library). It is not as fast as the original kernel written in C++ because the Lean compiler still has some ways to go to compete with C++ compilers, but the overhead is still reasonable in exchange for the additional guarantees it can provide, and it is suitable as a complementary step in Lean projects that wish to validate correctness in a variety of ways.

The second part of this project, which is much larger and on which we report partial progress, is the specification of Lean's metatheory and verification of the \lfl kernel with respect to that theory, covering mistake prevention measure (3). This is similar to the MetaRocq project \cite{sozeau20,coqcoqcorrect}, and amounts to a formalization of \cite{leantt}. Our principal contributions here are a definition of the typing judgment, some regularity theorems we can prove and others we can state, and the verification of some representative functions from the typechecker. The remainder is left as future work.

In this paper, we will avoid making value judgments about Lean or Rocq and their respective theories. These systems came to be what they are now as a result of complex historical factors, and we seek only to validate the behavior of Lean as it exists, even though it lacks some type theoretic properties that would have made the job easier if true. One possible outcome of this project is a recommendation for how to improve the typing rules, but this comes at a high cost to existing users so we do not take it lightly and we have no such recommendations currently.

The paper is organized as follows: \sref{Section}{sec:base-theory} defines the main typing judgment and describes some of its properties. \sref{Section}{sec:typeck} gives the core data structures of the typechecker and how they relate to the typing rules from \autoref{sec:base-theory}. \sref{Section}{sec:verify} sets the stage for the verification of the typechecker, and how it was used to find a soundness bug in the kernel. \sref{Section}{sec:env} explains how the global structure of the environment is put together from individual typing judgments. \sref{Section}{sec:inductive} explains the complexities associated with supporting inductive types, treated as an extension of the base MLTT theory. \sref{Section}{sec:results} gives some performance results regarding the complete verifier, and \autoref{sec:related-work} compares this project to MetaRocq, a Rocq formalization of Rocq metatheory and a verified kernel.

\section{Base theory}\label{sec:base-theory}

\subsection{Expressions}

At the heart of the theory is the \href{https://github.com/digama0/lean4lean/blob/cpp2026/Lean4Lean/Theory/VExpr.lean#L6}{\code{VExpr}}\footnote{The prefix ``\lstinline|V|'' for ``verified'' disambiguates it from the \href{https://github.com/leanprover/lean4/blob/v4.20.1/src/Lean/Expr.lean\#L301}{\code{Expr}} type used by the kernel itself. This is the more abstracted version of expressions used for specification purposes.} type, which represents expressions of type theory. This is a minimalistic version of the type actually used by the kernel in \sref{section}{sec:typeck-TrExpr}, to ease the burden of proving theoretical properties.\footnote{MetaRocq seems not to do this; it uses the same type for proving metatheory and for writing metaprograms (it's not the same type as the one Rocq itself uses because that type is in ML, but it is roughly equivalent).}

\begin{lstlisting}
inductive $\mbox{\href{https://github.com/digama0/lean4lean/blob/cpp2026/Lean4Lean/Theory/VExpr.lean\#L6}{VExpr}}$ where
  | bvar (deBruijnIndex : Nat)
  | sort (u : VLevel)
  | const (declName : Name) (us : List VLevel)
  | app (fn arg : VExpr)
  | lam (binderType body : VExpr)
  | forallE (binderType body : VExpr)
\end{lstlisting}

\begin{itemize}
  \item \lstinline|bvar $n$| represents the $n$'th variable in the context, counting from the inside out. As we will see, Lean itself uses ``locally nameless'' representation, with a combination of de Bruijn variables and named free variables, but here we only have pure de Bruijn variables.
  \item \lstinline|sort $u$| represents a universe, written as \lstinline|Sort $u$| (or \lstinline|Prop| and \lstinline|Type u| which are syntax for \lstinline|Sort 0| and \lstinline|Sort (u+1)| respectively).
  \item \lstinline|const $c$ $us$| represents a reference to global constant $c$ in the environment, instantiated with universes $us$. (Constants in Lean can be universe-polymorphic, and they are instantiated with concrete universes at each use site.) Inductive constructors, recursors, and definitions are all represented using this constructor.
  \item \lstinline|app $f$ $a$| is function application, written $f\;a$ or $f(a)$.
  \item \lstinline|lam $A$ $e$| is a lambda-abstraction $\lambda x:A.\, e$. Because we are using de Bruijn variables, $x$ is not represented; instead $e$ is typechecked in an extended context where variable $0$ is type $A$ and the other variables are shifted up by 1.
  \item \lstinline|forallE $A$ $B$| \footnote{The ``\lstinline|E|'' suffix is used because \lstinline|forall| is a keyword; we could use \lstinline|forall| anyway but it would require escaping in several places.} is a dependent $\Pi$-type $\Pi x:A.\, B$, also written as $\forall x:A.\, B$ because for \lstinline|Prop| this is the same as the ``for all'' quantifier. It uses the same binding structure as $\lambda$.
\end{itemize}

The \lstinline|sort $u$| and \lstinline|const $c$ $us$| constructors depend on another type \href{https://github.com/digama0/lean4lean/blob/cpp2026/Lean4Lean/Theory/VExpr.lean#L6}{\code{VLevel}}, which is the grammar of \emph{level expressions}:
\begin{lstlisting}
inductive $\mbox{\href{https://github.com/digama0/lean4lean/blob/cpp2026/Lean4Lean/Theory/VExpr.lean\#L6}{VLevel}}$ where
  | zero  : VLevel
  | succ  : VLevel → VLevel
  | max   : VLevel → VLevel → VLevel
  | imax  : VLevel → VLevel → VLevel
  -- imax a b means (if b = 0 then 0 else max a b)
  | param : Nat → VLevel
\end{lstlisting}
These come up when typechecking expressions, for example the type of \lstinline|sort u| is \lstinline|sort (succ u)|. The one notable case here is \lstinline|param $i$|, which represents the $i$th universe variable, where $i<n$ where $n$ is the number of universe parameters to the current declaration.

This definition corresponds to the following BNF-style grammar from \cite{leantt}, with the main change being that it is more precise about how variables and binding are represented. (For presentational reasons, we will use named variables in the informal version, but the formalization uses \lstinline|lift| and \lstinline|subst| functions when moving expressions between contexts.)
\begin{align*}
\ell&::=0\mid \suc\ell\mid \max(\ell,\ell)\mid \imax(\ell,\ell)\mid u\\
e&::=x\mid \U_\ell\mid c_{\bar{\ell}}\mid e\;e\mid \lambda x:e.\;e\mid \forall x:e.\;e\\
\Gamma&::=\cdot\mid \Gamma,x:e
\end{align*}

\subsection{Typing and definitional equality}
\begin{figure}
\begin{mathparpagebreakable}
\boxed{\Gamma\ni x:\alpha}\and
\axiom[l-zero]{\Gamma,x:\alpha\ni x:\alpha}\and
\infer[l-succ]{\Gamma\ni y:\beta}{\Gamma,x:\alpha\ni y:\beta}\\
\boxed{\Gamma\vdash_{E,n} e\equiv e':\alpha}\and
(\Gamma\vdash e:\alpha)\triangleq(\Gamma\vdash e\equiv e:\alpha)\and
\infer[t-bvar]{\Gamma\ni x:\alpha}{\Gamma\vdash x:\alpha}\and
\infer[t-symm]{\Gamma\vdash e\equiv e':\alpha}{\Gamma\vdash e'\equiv e:\alpha}\and
\infer[t-trans]{\Gamma\vdash e_1\equiv e_2:\alpha\quad\Gamma\vdash e_2\equiv e_3:\alpha}{\Gamma\vdash e_1\equiv e_3:\alpha}\and
\infer[t-sort]{n\vdash \ell,\ell'\ok\quad \ell\equiv\ell'}{\Gamma\vdash \U_\ell\equiv\U_{\ell'}:\U_{\suc\ell}}\and
\infer[t-const]{{\bar u}.(c_{\bar u}:\alpha)\in E\quad \forall i,\;n\vdash \ell_i,\ell'_i\ok\ \wedge\ \ell_i\equiv\ell'_i}{\Gamma\vdash c_{\bar{\ell}}\equiv c_{\bar{\ell'}}:\alpha[\bar{u}\mapsto \bar{\ell}]}\and
\infer[t-lam]{\Gamma\vdash\alpha\equiv \alpha':\U_{\ell}\quad \Gamma,x:\alpha\vdash e\equiv e':\beta}{\Gamma\vdash(\lambda x:\alpha.\;e)\equiv (\lambda x:\alpha'.\;e'):\forall x:\alpha.\;\beta}\and
\infer[t-all]{\Gamma\vdash \alpha\equiv \alpha':\U_{\ell_1}\quad \Gamma,x:\alpha\vdash \beta\equiv \beta':\U_{\ell_2}}{\Gamma\vdash (\forall x:\alpha.\;\beta)\equiv (\forall x:\alpha'.\;\beta'):\U_{\imax(\ell_1, \ell_2)}}\and
\infer[t-app]{\Gamma\vdash e_1\equiv e'_1:\forall x:\alpha.\;\beta\quad \Gamma\vdash e_2\equiv e'_2:\alpha}{\Gamma\vdash e_1\;e_2\equiv e'_1\;e'_2:\beta[x\mapsto e_2]}\and
\infer[t-conv]{\Gamma\vdash \alpha\equiv \beta:\U_{\ell}\quad \Gamma\vdash e\equiv e':\alpha}{\Gamma\vdash e\equiv e':\beta}\and
\infer[t-beta]{\Gamma,x:\alpha\vdash e:\beta\quad\Gamma\vdash e':\alpha}{\Gamma\vdash (\lambda x:\alpha.\;e)\;e'\equiv e[x\mapsto e']:\beta[x\mapsto e']}\and
\infer[t-eta]{\Gamma\vdash e:\forall y:\alpha.\;\beta}{\Gamma\vdash (\lambda x:\alpha.\;e\;x)\equiv e:\forall y:\alpha.\;\beta}\and
\infer[t-proof-irrel]{\Gamma\vdash p:\U_0\quad \Gamma\vdash h:p\quad \Gamma\vdash h':p}{\Gamma \vdash h\equiv h':p}\and
\infer[t-extra]{{\bar u}.(e\equiv e':\alpha)\in E\quad \forall i,\;n\vdash \ell_i\ok}{\Gamma\vdash e[\bar{u}\mapsto \bar{\ell}]\equiv e'[\bar{u}\mapsto \bar{\ell}]:\alpha[\bar{u}\mapsto \bar{\ell}]}
\end{mathparpagebreakable}
\caption{The rules for the judgment $\Gamma\vdash e\equiv e':\alpha$, with parameters $E$ (the global environment) and $n$ (the number of universe parameters in context), suppressed in the notation. This is \href{https://github.com/digama0/lean4lean/blob/cpp2026/Lean4Lean/Theory/Typing/Basic.lean\#L17}{\code{VEnv.IsDefEq}} in the formalization.}\label{fig:typing}
\end{figure}
For the typing rules, we use a slight variation of the typing judgments from \cite{leantt}, see \autoref{fig:typing}. Some remarks:

\begin{itemize}
  \item In \cite{leantt}, there were two judgments $\Gamma\vdash e:\alpha$ and $\Gamma\vdash e_1\equiv e_2$ which are mutually inductive (due to \textsc{t-conv} and \textsc{t-proof-irrel}), and $\Gamma\vdash e_1\equiv e_2:\alpha$ was a defined notion for $\Gamma\vdash e_1\equiv e_2\ \wedge\ \Gamma\vdash e_1:\alpha\ \wedge\ \Gamma\vdash e_2:\alpha$. In this version, we have only one all-in-one relation $\Gamma\vdash e_1\equiv e_2:\alpha$, and we define $\Gamma\vdash e:\alpha$ to mean $\Gamma\vdash e\equiv e:\alpha$ and define $\Gamma\vdash e_1\equiv e_2$ as $\exists\alpha.\;(\Gamma\vdash e_1\equiv e_2:\alpha)$.

  We conjecture the two formulations to be equivalent, but this version seems to be easier to prove basic structural properties about (see \autoref{sec:typing-properties}). Moreover Lean does not have good support for mutual inductive predicates, so keeping it single-inductive makes induction proofs easier.

  \item We also define $\Gamma\vdash \alpha\type$ to mean $\exists \ell.\ (\Gamma\vdash \alpha:\U_\ell)$.
  \item $\vdash\Gamma\ok$ is defined recursively, using $\Gamma\vdash A\type$.
  \item The rules \textsc{t-sort}, \textsc{t-const}, \textsc{t-extra} make use of a judgment $n\vdash \ell\ok$, which simply asserts that every \lstinline|param $i$| in $\ell$ satisfies $i<n$. The $\ell\equiv\ell'$ predicate asserts that $\ell$ and $\ell'$ are extensionally equivalent, i.e. for all assignments $v:\N\to\N$ of natural number values to the universe variables, $\scott{\ell}_{v}=\scott{\ell'}_{v}$.

  \item The \textsc{t-extra} rule says that we can add arbitrary definitional equalities from the environment. This is how we will add the two supported ``extensions'' to the theory, for inductive types and quotients, without the details of these extensions complicating reasoning about the core theory. But we can't actually support \emph{arbitrary} definitional equalities---instead each theorem about the typing judgment has its own assumptions about what extensions are allowed.  (However, note that the specific constraints on \textsc{t-extra} are still unfinished, pending future work on inductive type specification.)
\end{itemize}

\subsection{Properties of the typing judgment}\label{sec:typing-properties}

\begin{lemma}[basic properties]\label{lem:closedN}\ %
  \begin{enumerate}
    \item (\href{https://github.com/digama0/lean4lean/blob/cpp2026/Lean4Lean/Theory/Typing/Lemmas.lean#L321-L322}{\code{IsDefEq.closedN'}}) If $\Gamma\vdash e\equiv e':\alpha$ and $\Gamma$ is closed, then $e$, $e'$, and $\alpha$ are closed (all free variables have indices less than $|\Gamma|$).
    \item (\href{https://github.com/digama0/lean4lean/blob/cpp2026/Lean4Lean/Theory/Typing/Lemmas.lean#L485-L486}{\code{IsDefEq.weakN}}, weakening) If $\Gamma,\Gamma'\vdash e\equiv e':\alpha$, then $\Gamma,\Delta,\Gamma'\vdash e\equiv e':\alpha$.
    \item (\href{https://github.com/digama0/lean4lean/blob/cpp2026/Lean4Lean/Theory/Typing/Lemmas.lean#L554-L555}{\code{IsDefEq.instL}}) If $\Gamma\vdash_{E,n} e\equiv e':\alpha$ and $\forall i.\;n'\vdash \ell_i\ok$, then $\Gamma[\bar u\mapsto\bar\ell]\vdash_{E,n'} e[\bar u\mapsto\bar\ell]\equiv e'[\bar u\mapsto\bar\ell]:\alpha[\bar u\mapsto\bar\ell]$.
    \item (\href{https://github.com/digama0/lean4lean/blob/cpp2026/Lean4Lean/Theory/Typing/Lemmas.lean#L588-L589}{\code{IsDefEq.instN}}) If $\Gamma,x:\beta\vdash e_1\equiv e_2:\alpha$ and $\Gamma\vdash e_0:\beta$, then $\Gamma\vdash e_1[x\mapsto e_0]\equiv e_2[x\mapsto e_0]:\alpha[x\mapsto e_0]$.
  \end{enumerate}
\end{lemma}
\begin{proof}[Proof sketch]\ %
  \begin{enumerate}
    \item This is a straightforward proof by induction on $\Gamma\vdash e\equiv e':\alpha$, except that one gets stuck at \textsc{t-const} because (assuming the environment is well-typed, see \autoref{sec:env}) we know that $\vdash c_{\bar\ell}:\alpha$ and need that $\alpha$ is closed. So in fact this is a double induction, first over the environment $E$ and then over $\Gamma\vdash e\equiv e':\alpha$. We also need lemmas about $e[x\mapsto e']$ and $e[\bar u\mapsto\bar\ell]$ preserving closedness, but these are also direct by induction on $e$.
    \item By induction. We use \sref{Lemma}{lem:closedN}.1 in the \textsc{t-const} case, to show that the $\alpha$ in $\Gamma,\Gamma'\vdash c_{\bar\ell}\equiv c_{\bar\ell'}:\alpha$ can be used in $\Gamma,\Delta,\Gamma'\vdash c_{\bar\ell}\equiv c_{\bar\ell'}:\alpha$ without renaming any bound variables because $\alpha$ is closed.
    \item By induction.
    \item By induction on the first hypothesis. This uses\\ \sref{Lemma}{lem:closedN}.2 to lift $\Gamma\vdash e_0:\beta$ under binders.
    \vspace{-1.4em}
  \end{enumerate}
\end{proof}

\begin{lemma}[\href{https://github.com/digama0/lean4lean/blob/cpp2026/Lean4Lean/Theory/Typing/Lemmas.lean\#L648-L649}{\code{IsDefEq.defeqDF\_l}}]\label{lem:defeqDF_l}\ %
  If $\Gamma,x:\alpha\vdash e_1\equiv e_2:\beta$ and $\Gamma\vdash \alpha\equiv \alpha':\U_\ell$, then $\Gamma,x:\alpha'\vdash e_1\equiv e_2:\beta$.
\end{lemma}
\begin{proof}
  We have $\Gamma,x:\alpha'\vdash x:\alpha$ by \textsc{t-bvar}, \textsc{t-conv} and weakening, and $\Gamma,\hole:\alpha',x:\alpha\vdash e_1\equiv e_2:\beta$ by weakening, so $\Gamma,x:\alpha'\vdash e_1[x\mapsto x]\equiv e_2[x\mapsto x]:\beta[x\mapsto x]$ by \sref{Lemma}{lem:closedN}.4.
\end{proof}

\begin{lemma}[\href{https://github.com/digama0/lean4lean/blob/cpp2026/Lean4Lean/Theory/Typing/Lemmas.lean\#L725-L726}{\code{HasType.forallE\_inv}}]\label{lem:forallE_inv}\ %
  If $\Gamma\vdash (\forall x:\alpha.\;\beta):\gamma$ then $\Gamma\vdash \alpha\type$ and $\Gamma,x:\alpha\vdash \beta\type$ (and therefore also $\Gamma\vdash (\forall x:\alpha.\;\beta)\type$).
\end{lemma}
\begin{proof}[Proof sketch]
  By induction on $\Gamma\vdash e_1\equiv e_2:\gamma$, assuming that one of $e_1$ or $e_2$ is $\forall x:\alpha.\;\beta$. Most cases are trivial or inapplicable. Of those that remain:
  \begin{itemize}
    \item \textsc{t-all}: If $e_1$ is $\forall x:\alpha.\;\beta$ we are done; if $e_2$ is $\forall x:\alpha.\;\beta$ then we obtain $\Gamma,x:\alpha'\vdash \beta\type$ and need \sref{Lemma}{lem:defeqDF_l} to get $\Gamma,x:\alpha\vdash \beta\type$.
    \item \textsc{t-beta}: For this to apply, it must be that we have  $(\lambda y:\delta.\;e)\;e'\equiv e[y\mapsto e']$ where $e[y\mapsto e']$ is $\forall x:\alpha.\;\beta$. So either $e=y$ and $e'=\forall x:\alpha.\;\beta$, in which case the inductive hypothesis for $e'$ applies, or $e=\forall x:\alpha'.\;\beta'$ with $\alpha'[y\mapsto e']=\alpha$ and $\beta'[y\mapsto e']=\beta$ in which case $\Gamma,y:\delta\vdash \alpha'\type$ and $\Gamma,y:\delta,x:\alpha'\vdash \beta'\type$ by the inductive hypothesis for $e$, and \sref{Lemma}{lem:closedN}.4 applies.
    \item \textsc{t-extra}: It could be that $e[\bar u\mapsto\bar\ell]$ is $\forall x:\alpha.\;\beta$, but this can only happen if $e=\forall x:\alpha'.\;\beta'$ with $\alpha'[\bar u\mapsto\bar\ell]=\alpha$ and $\beta'[\bar u\mapsto\bar\ell]=\beta$, so by induction hypothesis (for the environment) $\vdash \alpha'\type$ and $x:\alpha'\vdash \beta'\type$, and we conclude using level substitution and weakening.
  \end{itemize}
\end{proof}

\begin{lemma}[\href{https://github.com/digama0/lean4lean/blob/cpp2026/Lean4Lean/Theory/Typing/Lemmas.lean\#L771}{\code{HasType.sort\_inv}}]\label{lem:sort_inv}\ %
  If $\Gamma\vdash_{E,n} \U_\ell:\gamma$ then $n\vdash \ell\ok$ (and therefore also $\Gamma\vdash \U_\ell:\U_{\suc\ell}$).
\end{lemma}
\begin{proof}[Proof sketch]
  Similar to \sref{Lemma}{lem:forallE_inv}.
\end{proof}

\begin{theorem}[\href{https://github.com/digama0/lean4lean/blob/cpp2026/Lean4Lean/Theory/Typing/Lemmas.lean\#L811-L812}{\code{IsDefEq.isType}}]\label{thm:isType}\ %
  If $\vdash\Gamma\ok$ and $\Gamma\vdash e:\alpha$ then $\Gamma\vdash \alpha\type$.
\end{theorem}
\begin{proof}[Proof sketch]
  By induction, using weakening and substitution lemmas. The nontrivial cases are:
  \begin{itemize}
    \item \textsc{t-app}: We have $\Gamma\vdash(\forall x:\alpha.\;\beta)\type$ from the IH and $\Gamma\vdash e_2:\alpha$ by assumption, and from \sref{Lemma}{lem:forallE_inv} we get $\Gamma,x:\alpha\vdash\beta\type$, so $\Gamma\vdash \beta[x\mapsto e_2]\type$ by the substitution lemma.
    \item \textsc{t-all}: We have $\Gamma\vdash\U_{\ell_1}\type$ and $\Gamma,x:\alpha\vdash\U_{\ell_1}\type$ from the IH, so by \sref{Lemma}{lem:sort_inv} we have $n\vdash \ell_1,\ell_2\ok$, therefore $\Gamma\vdash\U_{\imax(\ell_1, \ell_2)}\type$.
    \vspace{-1.4em}
  \end{itemize}
\end{proof}

\begin{lemma}[\href{https://github.com/digama0/lean4lean/blob/cpp2026/Lean4Lean/Theory/Typing/Lemmas.lean\#L818-L821}{\code{IsDefEq.instDF}}]\label{lem:instDF}\ %
  If $\Gamma,x:\alpha\vdash f\equiv f':\beta$ and $\Gamma\vdash a\equiv a':\alpha$ then $\Gamma\vdash f[x\mapsto a]\equiv f'[x\mapsto a']:\beta[x\mapsto a]$.
\end{lemma}
\begin{proof}[Proof sketch]
  This is a generalization of \sref{Lemma}{lem:closedN}.4, but we need \autoref{thm:isType} to prove that $\alpha$ and $\beta$ are well-typed to apply the rules below.

  First we show the claim assuming $\Gamma\vdash \beta[x\mapsto a]\equiv \beta[x\mapsto a']:\U_\ell$. In this case we have
\begin{align*}
  f[x\mapsto a]&\equiv (\lambda x:\alpha.\;f)\;a\\
  &\equiv (\lambda x:\alpha.\;f')\;a'\\
  &\equiv f'[x\mapsto a'] \qquad {}:\beta[x\mapsto a],
\end{align*}
where we use \textsc{t-beta} twice, and use the assumption $\beta[x\mapsto a]\equiv \beta[x\mapsto a']$ to justify the last step ecause \textsc{t-beta} gives the equality at the type $\beta[x\mapsto a']$ instead.

To finish, we apply the lemma twice, once with $f$ and $\beta$ so that it suffices to show $\beta[x\mapsto a]\equiv \beta[x\mapsto a']:\U_\ell$ and then again with $\beta$ in place of $f$ and $\U_\ell$ in place of $\beta$ so that it suffices to show $\U_\ell[x\mapsto a]\equiv \U_\ell[x\mapsto a']:\U_{\suc\ell}$, which is true by reflexivity.
\end{proof}

\subsection{Conjectured properties of the typing judgment}\label{sec:typing-conjectures}

Unfortunately, there are more structural properties of the typing judgment we need beyond the results in the previous section. Many of these are close relatives of each other and can be proved from other conjectures in this section.

\begin{conjecture}[\href{https://github.com/digama0/lean4lean/blob/cpp2026/Lean4Lean/Theory/Typing/UniqueTyping.lean\#L11-L13}{\code{IsDefEq.uniq}}]\label{conj:IsDefEq.uniq}
If $\Gamma\vdash e:\alpha$ and $\Gamma\vdash e:\beta$, then $\Gamma\vdash \alpha\equiv\beta:\U_\ell$ for some $\ell$.
\end{conjecture}
Unique typing has a major simplifying effect on the theory. Here are some consequences (with proofs omitted, but they are simple algebraic consequences of earlier lemmas):

\begin{corollary}[Consequences of unique typing]\ %
  \begin{enumerate}
    \item (\href{https://github.com/digama0/lean4lean/blob/cpp2026/Lean4Lean/Theory/Typing/UniqueTyping.lean#L69-L71}{\code{IsDefEqU.trans}}) If $\Gamma\vdash e_1\equiv e_2$ and $\Gamma\vdash e_2\equiv e_3$, then $\Gamma\vdash e_1\equiv e_3$. (This is \textsc{t-trans} but without requiring the types to match.)
    \item (\href{https://github.com/digama0/lean4lean/blob/cpp2026/Lean4Lean/Theory/Typing/UniqueTyping.lean#L23-L25}{\code{isDefEq\_iff}}) $\Gamma\vdash e\equiv e':\alpha$ if and only if $\Gamma\vdash e:\alpha$, $\Gamma\vdash e':\alpha$, and $\Gamma\vdash e\equiv e'$. (Hence this formulation implies the one of \cite{leantt}.)
    \item (\href{https://github.com/digama0/lean4lean/blob/cpp2026/Lean4Lean/Theory/Typing/UniqueTyping.lean#L42-L44}{\code{IsDefEqU.defeqDF}}) If $\Gamma\vdash e\equiv e':\alpha$ and $\Gamma\vdash \alpha\equiv\beta$, then $\Gamma\vdash e\equiv e':\beta$. (This is \textsc{t-conv} but with an untyped definitional equality.)
  \end{enumerate}
\end{corollary}

\noindent This group of theorems is proved mutually with \sref{Conjecture}{conj:IsDefEq.uniq} in \cite{leantt}:
\begin{conjecture}[Definitional inversion]\label{conj:IsDefEqU.forallE_inv}\ %
  \begin{enumerate}
    \item (\href{https://github.com/digama0/lean4lean/blob/cpp2026/Lean4Lean/Theory/Typing/UniqueTyping.lean#L73-L74}{\code{IsDefEqU.sort\_inv}}) If $\Gamma\vdash \U_\ell\equiv\U_{\ell'}$ then $\ell\equiv \ell'$.
    \item (\href{https://github.com/digama0/lean4lean/blob/cpp2026/Lean4Lean/Theory/Typing/UniqueTyping.lean#L76-L78}{\code{IsDefEqU.forallE\_inv}}) If $\Gamma\vdash (\forall x:\alpha.\;\beta)\equiv(\forall x:\alpha'.\;\beta')$ then $\Gamma\vdash\alpha\equiv \alpha'$ and $\Gamma,x:\alpha\vdash\beta\equiv \beta'$.\footnote{This theorem is sometimes called ``injectivity of $\forall$'' but this name is slightly misleading because $\forall$ is not injective as a function on types, which is to say if we use propositional equality instead of def.eq then the theorem fails.}
    \item (\href{https://github.com/digama0/lean4lean/blob/cpp2026/Lean4Lean/Theory/Typing/UniqueTyping.lean#L80-L81}{\code{IsDefEqU.sort\_forallE\_inv}}) $\Gamma\vdash \U_\ell\not\equiv(\forall x:\alpha.\;\beta)$.
  \end{enumerate}
\end{conjecture}

The reason \sref{Conjecture}{conj:IsDefEq.uniq} and \sref{Conjecture}{conj:IsDefEqU.forallE_inv} have been downgraded from theorems in \cite{leantt} to conjectures here is because the proof has an error in one of the technical lemmas, and it remains to be seen if it is possible to salvage the proof. More precisely, the proof constructs a stratification $\vdash_i$ of the typing judgment in order to break the mutual induction between typing and definitional equality, but this stratification does not and cannot respect substitution; that is, if $\Gamma,x:\beta\vdash_i e:\alpha$ and $\Gamma\vdash_j e':\beta$, then the proof requires $\Gamma\vdash_{\max(i,j)} e[x\mapsto e']:\alpha$ but only $\Gamma\vdash_{i+j} e[x\mapsto e']:\alpha$ holds.

We still have reasons to believe the conjectures here are true, but more work is needed to determine how to structure the proof by induction. In \cite{lennonbertrand2025}, Lennon-Bertrand explores some of the landscape of proof approaches that could be applied here, based on similar work for MetaRocq. The proof of soundness is not impacted because there are alternative routes to construct the model that avoid unique typing (also described in \cite{leantt}), but these conjectures are necessary in at least some form in order to prove the correctness of the typechecker (see \autoref{sec:typeck}).

Another class of conjectured theorems concerns the ``invertibility'' of weakening:\footnote{This is not a trivial theorem, and MetaRocq found a bug in Rocq here; see section 3.5 of \cite{sozeau-touring}.}
\begin{conjecture}[\href{https://github.com/digama0/lean4lean/blob/cpp2026/Lean4Lean/Theory/Typing/UniqueTyping.lean\#L84-L85}{Strengthening}]\label{conj:IsDefEqU.weakN_inv}\ %
  If $e,e'$ do not mention any variables in $\Delta$, then $\Gamma,\Delta,\Gamma'\vdash e\equiv e'$ if and only if $\Gamma,\Gamma'\vdash e\equiv e'$.
\end{conjecture}
\begin{corollary}[Consequences of strengthening]\ %
  \begin{enumerate}
    \item (\href{https://github.com/digama0/lean4lean/blob/cpp2026/Lean4Lean/Theory/Typing/UniqueTyping.lean#L92-L95}{\code{IsDefEq.skips}}) If $\Gamma,\Delta,\Gamma'\vdash e\equiv e':\alpha$ and $e,e'$ do not mention the variables in $\Delta$, then there exists $\alpha'$ not mentioning $\Delta$ such that $\Gamma,\Delta,\Gamma'\vdash e\equiv e':\alpha'$.
    \item (\href{https://github.com/digama0/lean4lean/blob/cpp2026/Lean4Lean/Theory/Typing/UniqueTyping.lean#L123-L124}{\code{IsDefEq.weakN\_iff}}) If $e,e',\alpha$ do not mention any variables in $\Delta$, then $\Gamma,\Delta,\Gamma'\vdash e\equiv e':\alpha$ if and only if $\Gamma,\Gamma'\vdash e\equiv e':\alpha$.
    \item (\href{https://github.com/digama0/lean4lean/blob/cpp2026/Lean4Lean/Theory/Typing/UniqueTyping.lean#L110-L111}{\code{OnCtx.weakN\_inv}}) If $\vdash\Gamma,\Delta,\Gamma'\ok$ and $\Gamma'$ does not mention the variables in $\Delta$ then $\vdash\Gamma,\Gamma'\ok$.
  \end{enumerate}
\end{corollary}

\section{The typechecker}\label{sec:typeck}
\subsection{Relating VExpr to Expr}\label{sec:typeck-TrExpr}

The actual Lean kernel does not use any of the machinery from the previous section directly. Instead, it works with another type, \href{https://github.com/leanprover/lean4/blob/v4.20.1/src/Lean/Expr.lean#L301}{\code{Expr}}:
\begin{lstlisting}
inductive $\mbox{\href{https://github.com/leanprover/lean4/blob/v4.20.1/src/Lean/Expr.lean\#L301}{Expr}}$ where
  | bvar (deBruijnIndex : Nat)
  | fvar (fvarId : FVarId)
  | mvar (mvarId : MVarId)
  | sort (u : Level)
  | const (declName : Name) (us : List Level)
  | app (fn arg : Expr)
  | lam (binderName : Name) (type body : Expr)
        (info : BinderInfo)
  | forallE (binderName : Name) (type body : Expr)
            (info : BinderInfo)
  | letE (declName : Name) (type value body : Expr)
         (nonDep : Bool)
  | lit : Literal → Expr
  | mdata (data : MData) (expr : Expr)
  | proj (typeName : Name) (idx : Nat)
         (struct : Expr)
\end{lstlisting}
Comparing this with \href{https://github.com/digama0/lean4lean/blob/cpp2026/Lean4Lean/Theory/VExpr.lean#L6}{\code{VExpr}} reveals some differences:
\begin{itemize}
  \item \lstinline|mvar $n$| is for metavariables in the elaborator. They are banned in kernel terms.
  \item \lstinline|fvar $n$| is a ``free variable'', which is identified by a name rather than an index into the context. Although these are also banned, the kernel creates \lstinline|fvar $n$| expressions during typechecking, and converting between locally nameless and pure de Bruijn is one of the major differences we will have to deal with in verification of the kernel.
  \item The \href{https://github.com/leanprover/lean4/blob/v4.20.1/src/Lean/Level.lean#L90}{\code{Level}} type (not shown) also differs from \href{https://github.com/digama0/lean4lean/blob/cpp2026/Lean4Lean/Theory/VExpr.lean#L6}{\code{VLevel}} similarly.
  \item The \lstinline|lam| and \lstinline|forallE| constructors contain additional elaborator metadata.
  \item The \lstinline|letE| construct is used for \lstinline|let x := e1; e2| expressions. We handle these by simply expanding them to $\code{e2}[\code{x}\mapsto \code{e1}]$.
  \item The \lstinline|lit| constructor is for literals, \lstinline|Literal := Nat ⊕ String|.
  This is one of the new features of the Lean 4 kernel: natural numbers are represented directly as bignums and several functions like \href{https://github.com/leanprover/lean4/blob/v4.20.1/src/Init/Prelude.lean#L1658-L1667}{\code{Nat.mul}} are overridden with a native implementation rather than using the inductive structure of \href{https://github.com/leanprover/lean4/blob/v4.20.1/src/Init/Prelude.lean#L1164}{\code{Nat}} directly (effectively writing numerals in unary). For the purpose of modeling in \href{https://github.com/digama0/lean4lean/blob/cpp2026/Lean4Lean/Theory/VExpr.lean#L6}{\code{VExpr}} however, we just perform exactly this (exponential-size) unfolding in terms of \lstinline|Nat.zero| and \lstinline|Nat.succ| applications. Similarly, \lstinline|String| is a wrapper around \lstinline|List Char| for the purpose of modeling, even though the actual representation in the runtime is as a UTF-8 encoded byte string.

  \item The \lstinline|mdata $d$ $e$| constructor is equivalent to $e$ for modeling purposes (this is used for associating metadata to expressions).
  \item \lstinline|proj $c$ $i$ $e$| is a projection out of a \lstinline|structure| or struc\-ture-like inductive type. (See \autoref{sec:inductive}.) This is the most challenging constructor to desugar, because even though it is equivalent to an application of the recursor for the inductive type, we need the type of $e$ to construct this term. This means that the desugaring function must be type-aware, and is only well-defined up to definitional equality.
\end{itemize}

It is worth mentioning that \href{https://github.com/leanprover/lean4/blob/v4.20.1/src/Lean/Expr.lean#L301}{\code{Expr}} was not defined as part of this project --- this is a type imported directly from the Lean elaborator. In fact, the C++ Lean kernel uses FFI (foreign-function interface) to Lean in order to make use of this type, so \lfl and the C++ kernel are sharing this data structure.

As a result of these considerations, we use an inductive type to encapsulate the translation from \href{https://github.com/leanprover/lean4/blob/v4.20.1/src/Lean/Expr.lean#L301}{\code{Expr}} to \href{https://github.com/digama0/lean4lean/blob/cpp2026/Lean4Lean/Theory/VExpr.lean#L6}{\code{VExpr}}:
\begin{lstlisting}
inductive $\mbox{\href{https://github.com/digama0/lean4lean/blob/cpp2026/Lean4Lean/Verify/Typing/Expr.lean\#L47}{TrExpr}}$ (env : VEnv) (Us : List Name) :
    VLCtx → Expr → VExpr → Prop
\end{lstlisting}
Here $\mathit{env};\mathit{Us};\Delta\vdash e\leadsto e'$ asserts that in local context $\Delta$, with names $\mathit{Us}$ for the local universe parameters and $\mathit{env}$ for the global environment (see \autoref{sec:env}), $e:\code{VExpr}$ is a translation of $e':\code{Expr}$. It also asserts that $\Delta$ is a well-typed context and $e'$ is well-typed in it, so we can use $\exists e'.\;\mathit{env};\mathit{Us};\Delta\vdash e\leadsto e'$ to assert that $e$ is a well typed \href{https://github.com/leanprover/lean4/blob/v4.20.1/src/Lean/Expr.lean#L301}{\code{Expr}}, which is the target specification we have for correctness of the typechecker.

The type \href{https://github.com/digama0/lean4lean/blob/cpp2026/Lean4Lean/Verify/VLCtx.lean#L40}{\code{VLCtx}} appearing here is new, and it contains the information we need to translate between the \href{https://github.com/leanprover/lean4/blob/v4.20.1/src/Lean/Expr.lean#L301}{\code{Expr}} and \href{https://github.com/digama0/lean4lean/blob/cpp2026/Lean4Lean/Theory/VExpr.lean#L6}{\code{VExpr}} types:

\begin{lstlisting}
inductive $\mbox{\href{https://github.com/digama0/lean4lean/blob/cpp2026/Lean4Lean/Verify/VLCtx.lean\#L8}{VLocalDecl}}$ where
  | vlam (type : VExpr)         -- x : type
  | vlet (type value : VExpr)   -- x : type := value

def $\mbox{\href{https://github.com/digama0/lean4lean/blob/cpp2026/Lean4Lean/Verify/VLCtx.lean\#L40}{VLCtx}}$ := List (Option FVarId × VLocalDecl)
\end{lstlisting}

The idea is that \lstinline|vlam $\alpha$| represents a regular variable in the context, the only kind of variable we had in \autoref{sec:base-theory}, while \lstinline|vlet $\alpha$ $v$| is used when we are inside the context of a \lstinline|let x : $\alpha$ := $v$; _| term. We can build a local context in the sense of \autoref{sec:base-theory} by simply dropping the \lstinline|vlet $\alpha$ $v$| terms:

\begin{lstlisting}
def $\mbox{\href{https://github.com/digama0/lean4lean/blob/cpp2026/Lean4Lean/Verify/VLCtx.lean\#L85}{toCtx}}$ : VLCtx → List VExpr
  | [] => []
  | (_, vlam α) :: Δ => α :: toCtx Δ
  | (_, vlet _ _) :: Δ => toCtx Δ
\end{lstlisting}
We have to be careful to reindex the \lstinline|Expr.bvar $i$| indices though, since this includes both let- and lambda-variables while \lstinline|VExpr.bvar $i$| only counts the lambda-variables, with the let-variables expanded to terms.

The \lstinline|Option| \href{https://github.com/leanprover/lean4/blob/v4.20.1/src/Lean/Expr.lean#L227}{\code{FVarId}} in the context represents the ``name'' of the variables. Lean follows the ``locally nameless'' discipline, in which most operations act only on ``closed terms'' (where closed means that there are no \lstinline|bvar $i$| variables outside a binder), so any time it needs to process a term under a binder it first \emph{instantiates} the term, replacing \lstinline|bvar $(i+1)$| with \lstinline|bvar $i$| and \lstinline|bvar 0| with \lstinline|fvar $a$|, where $a:\mbox{\href{https://github.com/leanprover/lean4/blob/v4.20.1/src/Lean/Expr.lean\#L227}{\code{FVarId}}}$ is a (globally) fresh variable name. We store these names in $\Delta$ so that we know how to relate them to de Bruijn indices in the context. The name generator uses a counter in the global state of the typechecker, and there is an invariant that all \href{https://github.com/leanprover/lean4/blob/v4.20.1/src/Lean/Expr.lean#L227}{\code{FVarId}}s in the state are less than the counter value.

The reason it is an \lstinline|Option| is because when we are typing an \emph{open} term, we still have to deal with variables that have not been assigned \href{https://github.com/leanprover/lean4/blob/v4.20.1/src/Lean/Expr.lean#L227}{\code{FVarId}}s. Even though Lean does not directly handle such terms, they appear as subterms of expressions that are handled so we need both cases to have a compositional specification.

\subsection{The typechecker implementation}\label{sec:typeck-impl}
The kernel itself is defined completely separately from \href{https://github.com/digama0/lean4lean/blob/cpp2026/Lean4Lean/Theory/VExpr.lean#L6}{\code{VExpr}}, and is a close mirror of the C++ code, modulo translation into functional style. The two implementations do essentially the same things in the same order, and as a result, most bugs or peculiarities we have discovered in the course of formalizing \lfl are replicated in the C++ kernel, which gives us hope that verifying \lfl will lend strong evidence to the correctness of the C++ code as well.

\subsubsection{Untangling recursion via RecM}\label{sec:RecM}
The kernel consists of a large mutually-recursive set of functions. Because this is difficult for Lean to handle, and it also causes issues for proving theorems about the theorems, it is desirable to break the dependencies.
It turns out that the call graph's feedback vertex set (FVS) number is 4, which is to say that if the below 4 functions are each treated as two nodes with in-edges separate from out-edges, then the call graph becomes acyclic (except for self-loops corresponding to normal recursion) and we can define them in dependency order. Conveniently, all of these functions are also important high-level entry points:

\begin{enumerate}
  \item \href{https://github.com/digama0/lean4lean/blob/cpp2026/Lean4Lean/TypeChecker.lean#L665}{\code{isDefEqCore}} $s$ $t$ returns \lstinline|true| if $\Gamma\vdash s\equiv t$ is provable. (This is normally used via \href{https://github.com/digama0/lean4lean/blob/cpp2026/Lean4Lean/TypeChecker.lean#L145}{\code{isDefEq}} $s$ $t$, which also caches this result, but some callers use \href{https://github.com/digama0/lean4lean/blob/cpp2026/Lean4Lean/TypeChecker.lean#L665}{\code{isDefEqCore}} directly.)
  \item \href{https://github.com/digama0/lean4lean/blob/cpp2026/Lean4Lean/TypeChecker.lean#L435}{\code{whnf}} $e$ returns the weak head normal form (WHNF) of $e$. From a modeling perspective, the main important property is that if it returns $e'$ and $e$ is well-typed then $\Gamma\vdash e\equiv e'$ is provable.
  \item \href{https://github.com/digama0/lean4lean/blob/cpp2026/Lean4Lean/TypeChecker.lean#L310}{\code{whnfCore}} $e$ \lstinline|cheapRec| \lstinline|cheapProj| has the same specification as \href{https://github.com/digama0/lean4lean/blob/cpp2026/Lean4Lean/TypeChecker.lean#L435}{\code{whnf}} $e$, but it lacks some early exit paths. \lstinline|cheapRec| and \lstinline|cheapProj| are flags affecting which kinds of terms are unfolded. The reason \href{https://github.com/digama0/lean4lean/blob/cpp2026/Lean4Lean/TypeChecker.lean#L435}{\code{whnf}} and \href{https://github.com/digama0/lean4lean/blob/cpp2026/Lean4Lean/TypeChecker.lean\#L310}{\code{whnfCore}} show up separately in this list is because they are both used in other functions (i.e. some functions call \href{https://github.com/digama0/lean4lean/blob/cpp2026/Lean4Lean/TypeChecker.lean#L310}{\code{whnfCore}} directly because e.g. the fast path isn't applicable).
  \item \href{https://github.com/digama0/lean4lean/blob/cpp2026/Lean4Lean/TypeChecker.lean#L240}{\code{inferType}} $e$ \lstinline|inferOnly| infers or typechecks a term. That is, if \href{https://github.com/digama0/lean4lean/blob/cpp2026/Lean4Lean/TypeChecker.lean#L240}{\code{inferType}} $e$ \lstinline|false| returns $\alpha$ then $\Gamma\vdash e:\alpha$ holds, and if \href{https://github.com/digama0/lean4lean/blob/cpp2026/Lean4Lean/TypeChecker.lean#L240}{\code{inferType}} $e$ \lstinline|true| returns $\alpha$ then $\Gamma\vdash e:\alpha$ holds assuming $\Gamma\vdash e:\alpha'$ for some $\alpha'$.
\end{enumerate}

\noindent We use this to define the \href{https://github.com/digama0/lean4lean/blob/cpp2026/Lean4Lean/TypeChecker.lean#L59}{\code{RecM}} monad:

\begin{lstlisting}
structure $\mbox{\href{https://github.com/digama0/lean4lean/blob/cpp2026/Lean4Lean/TypeChecker.lean\#L53}{Methods}}$ where
  isDefEqCore : Expr → Expr → M Bool
  whnfCore (e : Expr) (cheapRec := false)
    (cheapProj := false) : M Expr
  whnf (e : Expr) : M Expr
  inferType (e : Expr) (inferOnly : Bool) : M Expr

abbrev $\mbox{\href{https://github.com/digama0/lean4lean/blob/cpp2026/Lean4Lean/TypeChecker.lean\#L59}{RecM}}$ := ReaderT $\mbox{\href{https://github.com/digama0/lean4lean/blob/cpp2026/Lean4Lean/TypeChecker.lean\#L53}{Methods}}$ $\mbox{\href{https://github.com/digama0/lean4lean/blob/cpp2026/Lean4Lean/TypeChecker.lean\#L32}{M}}$
\end{lstlisting}
This is a monad stack on top of the main \href{https://github.com/digama0/lean4lean/blob/cpp2026/Lean4Lean/TypeChecker.lean#L32}{\code{TypeChecker.M}} monad which contains the actual state of the typechecker.


Now, this method of calling functions to avoid recursion only kicks the problem up one step. How do we construct an element of \href{https://github.com/digama0/lean4lean/blob/cpp2026/Lean4Lean/TypeChecker.lean#L53}{\code{Methods}} if all the functions require another \href{https://github.com/digama0/lean4lean/blob/cpp2026/Lean4Lean/TypeChecker.lean#L53}{\code{Methods}}? Ideally, we would actually prove the termination of the kernel, because DTT is supposed to be terminating. However:
\begin{itemize}
  \item It is unlikely that we can prove termination of a typechecker for Lean in Lean, because although the soundness proof from \cite{leantt} does not depend on termination, MetaRocq's does \cite{coqcoqcorrect}, and generally termination measures for DTT require large cardinals of comparable strength to the proof theory. We are up against G\"{o}del's incompleteness theorem, so anything that would imply the unconditional soundness of Lean won't be directly provable.
  \item Besides this, the Lean type theory is known not to terminate. Coquand and Abel \cite{coqabel} constructed a counterexample to strong normalization using reduction of proofs, and this can be shown to impact definitional equality checks even for regular types:
\begin{lstlisting}
/-! Andreas-Abel construction of
    nontermination in proofs -/
def True' := ∀ p : Prop, p → p
def om : True' := fun A a =>
  @cast (True' → True') A
    (propext ⟨fun _ => a, fun _ => id⟩)
    (fun z => z (True' → True') id z)
def Om : True' := om (True' → True') id om
#reduce Om -- whnf nontermination

/-! nontermination outside proofs: -/
inductive Foo : Prop | mk : True' → Foo
def foo : Foo := Om _ (Foo.mk fun _ => id)
example : foo.recOn (fun _ => 1) = 1 := by
  rfl -- isDefEq nontermination
\end{lstlisting}
  Essentially, this is a combination of impredicativity, proof irrelevance, and subsingleton elimination.
  \item The kernel does not loop forever in many cases because this is a bad user experience --- it has timeouts and depth limits. Not all parts of the kernel have such limits, but it does give us a reasonable design principle which fortuitously solves our termination problem.
\end{itemize}

So we use what is arguably%
\footnote{There are alternative methods for representing partial functions such as the Bove-Capretta method (take a proof of termination as argument), but we do not want to use this as the kernel API is fixed to match Lean's actual kernel function, and we also actually want it to time out on extreme cases.}
the standard solution for defining partial functions in a language like Lean or Rocq: use a fuel parameter, a natural number which counts the number of nested recursive calls to one of the \href{https://github.com/digama0/lean4lean/blob/cpp2026/Lean4Lean/TypeChecker.lean#L53}{\code{Methods}}, and throw a \href{https://github.com/leanprover/lean4/blob/v4.20.1/src/Lean/Environment.lean#L267}{\code{deepRecursion}} error if we run out of fuel.
Currently, this limit is a fixed constant (1000), which turns out to be sufficient for checking all of \textsf{Mathlib}, but this could be made configurable.\footnote{This is not a limit on the depth of expressions exactly, these use structural recursion and hence need no fuel for the termination argument; instead fuel is only consumed when making a recursive call that is not otherwise decreasing, for example when reducing definitions to WHNF. The code that is most likely to hit depth limits is in proofs by reflection, but these are comparatively rare, in part because the kernel algorithm for this is not very efficient.}

\subsubsection{Inferring types}\label{sec:inferType}
The kernel's typechecking algorithm is fairly straightforward, as far as DTT typecheckers go. It does not do bidirectional typechecking; instead it simply propagates forward the type of expressions via the \href{https://github.com/digama0/lean4lean/blob/cpp2026/Lean4Lean/TypeChecker.lean#L240}{\code{inferType}} function, which has an \lstinline|inferOnly| flag which is true if we can assume the expression is well-typed. For example, here is how $\lambda x:\alpha.\;e$ expressions are typed:
\begin{lstlisting}
def $\mbox{\href{https://github.com/digama0/lean4lean/blob/cpp2026/Lean4Lean/TypeChecker.lean\#L115}{inferLambda}}$ (e : Expr) (inferOnly : Bool) :
    RecM Expr := loop #[] e where
  loop fvars : Expr → RecM Expr
  | .lam name dom body bi => do
    let d := dom.instantiateRev fvars
    if !inferOnly then
      let sort ← inferType d inferOnly
      _ ← ensureSortCore sort d
    withLocalDecl name d bi fun fv => do
      loop (fvars.push fv) body
  | e => do
    let e' := e.instantiateRev fvars
    let r ← inferType e' inferOnly
    let r := r.cheapBetaReduce
    return (← getLCtx).mkForall fvars r
\end{lstlisting}

\noindent Observe that this does more than simply checking \textsc{t-lam}:
\begin{itemize}
  \item Since it is using locally-nameless representation, whenever it needs to enter a binder it has to create a fresh \code{fvar} and replace \lstinline|bvar 0|. The \href{https://github.com/digama0/lean4lean/blob/cpp2026/Lean4Lean/LocalContext.lean#L13}{\code{withLocalDecl}} function does this, and changes the local context to include the new \code{fvar}, but instantiation is deferred here.
  \item It actually consumes a \emph{sequence} of lambdas, using the tail-recursive function \code{loop}. This is a performance optimization, since instantiation is $O(n)$ and so repeatedly instantiating the body in each step of the loop would cause $O(n^2)$ work. The actual instantiation is the two calls to \href{https://github.com/leanprover/lean4/blob/v4.20.1/src/Lean/Expr.lean#L1400}{\code{instantiateRev}}.
  \item When \code{inferOnly} is false, meaning that it can assume the term is well-formed, it does not need to typecheck the type of the lambda.
  \item At the end of the loop, it uses \href{https://github.com/digama0/lean4lean/blob/cpp2026/Lean4Lean/Instantiate.lean#L8}{\code{cheapBetaReduce}} to reduce the result type in the common case that it has the form $(\lambda \overline{x}.\;x_i)\;\overline{e}$, or $(\lambda \overline{x}.\;c)\;\overline{e}$ (where $c$ does not depend on $\overline{x}$).
\end{itemize}
We will discuss this function further in \autoref{sec:verify}.

\subsubsection{WHNF and reduction}
The \href{https://github.com/digama0/lean4lean/blob/cpp2026/Lean4Lean/TypeChecker.lean#L435}{\code{whnf}} function performs reductions at the ``head'' of an expression until a constructor is exposed. In so doing it implements a lazy, call-by-name evaluation strategy.

\begin{itemize}
  \item A variable, sort, axiom, lambda, forall, or literal is in WHNF.
  \item An application $f\;a$ reduces $f$ to WHNF first. If it reduces to $\lambda x.\;e$, then $\beta$-reduce to $e[a/x]$ and continue.
  \item A let-expression $\mathsf{let}\;x:=a;\;e$ reduces to $e[a/x]$ and continues.
  \item The recursor of an inductive type reduces its major argument first. If it reduces to a constructor, then $\beta$-reduce (e.g. $\mathsf{rec}_\N\;C\;z\;s\;(\mathsf{suc}\;n)\rightsquigarrow s\;n\;(\mathsf{rec}_\N\;C\;z\;s\;n)$) and continue.
  \item A definition reduces to its body and continues.
\end{itemize}

\noindent Beyond these basics, there are a few additional Lean specifics:
\begin{itemize}
  \item There is a list of 15 built-in functions on $\N$ like $+$, $-$, $\times$, $\mathsf{div}$, $\mathsf{mod}$, $\mathsf{bitAnd}$, etc., which are evaluated in call-by-value order. For example:
  \begin{itemize}
    \item If we are reducing $a+b$ to WHNF, then first reduce $a$ and $b$. If these reduce to literals $[n]$, $[m]$, then compute $[n+m]$ and stop. (Compare: the default behavior on $a+1$ would reduce to $\mathsf{suc}\;(a+0)$ without reducing $a$.)
  \end{itemize}
  \item When encountering $\mathsf{\href{https://github.com/leanprover/lean4/blob/v4.20.1/src/Init/Core.lean\#L2344-L2363}{reduceBool}}\;c$ where $c:\mathsf{Bool}$ is a constant definition, $c$ is compiled to interpreter bytecode and executed. The result (\lstinline|true| or \lstinline|false|) is reified as an expression of type $\mathsf{Bool}$ as the result.
\end{itemize}

Supporting built-in functions on $\N$ in \lfl was not difficult, but doing so soundly required implementing checks to ensure that e.g. when a definition named \href{https://github.com/leanprover/lean4/blob/v4.20.1/src/Init/Prelude.lean#L1639-L1648}{\code{Nat.add}} is added to the environment, it satisfies the definitional equalities $a+0\equiv a$ and $a+\mathsf{suc}\;n\equiv \mathsf{suc}\;(a+n)$. This allows us to prove by induction that this definition will compute on literals such that $[n]+[m]\equiv[n+m]$.\footnote{The original Lean kernel does not perform these checks. This is technically unsound, but to exploit it one would need to replace the prelude, and this is not a supported configuration.}

On the other hand, supporting native evaluation using $\mathsf{\href{https://github.com/leanprover/lean4/blob/v4.20.1/src/Init/Core.lean\#L2344-L2363}{reduceBool}}$ is not nearly so easy. In theory, with a formalization of the IR (intermediate representation) semantics, it may be possible to prove the soundness of IR compilation and interpretation. However, IR can also call external functions from C, and with \href{https://github.com/leanprover/lean4/blob/v4.20.1/src/Lean/Compiler/ImplementedByAttr.lean#L14}{\code{implemented\_by}} one can completely decouple the compiler behavior from the kernel's version of the implementation, so the general feature is simply ``unsound by design.'' Given that \textsf{Mathlib} and many other mathematics projects avoid \href{https://github.com/leanprover/lean4/blob/v4.20.1/src/Init/Core.lean#L2344-L2363}{\textsf{reduceBool}} completely for these reasons, we do not lose much by not supporting it.

\subsubsection{Definitional equality checking}
When writing a DTT type checker, this is the least ``canonical'' part. There are many different heuristics one can employ here, and because the worst-case time complexity is galactically large, the heuristics are critical for ensuring that typechecking stays within the same ballpark as when it was first checked in the presence of the human author (since we are only aiming to be able to check Lean projects in the wild). The algorithm for checking $s\overset{?}{\equiv} t$ in Lean goes roughly as follows:

\begin{itemize}
  \item If $s$ and $t$ are both lambdas, sorts, foralls, or literals, then compare subterms.
  \item Validate $\mathsf{true}\equiv t$ if $t\rightsquigarrow \mathsf{true}$.\footnote{This is an optimization for proof by reflection, where a goal $t=\mathsf{true}$ is proved using $\mathsf{refl}\;\mathsf{true}:\mathsf{true}=\mathsf{true}$.}
  \item Reduce $s\rightsquigarrow s'$ and $t\rightsquigarrow t'$ with \lstinline|cheapProj := true|; continue with $s'\overset{?}{\equiv} t'$.
  \item If $s:\alpha$ and $t:\beta$ and $\alpha\equiv\beta:\U_0$ then return true (\textsc{t-proof-irrel}).
  \item If it is $f\;\overline{a}\overset{?}{\equiv} g\;\overline{b}$ where both $f$ and $g$ are definitions, if $\mathsf{height}(f)<\mathsf{height}(g)$ then unfold $g$ and continue; else unfold $f$ and continue. Here $\mathsf{height}(f)$ is an unfolding heuristic based on definitional height (where $\mathsf{height}(f)=1+\max_{c\preceq e}\mathsf{height}(c)$;)
  \item Compare subterms for applications, projections, and variables.
  \item Reduce $s\rightsquigarrow s'$ and $t\rightsquigarrow t'$ again, if this makes progress.
  \item If it is $f\;\overline{a}\overset{?}{\equiv} f\;\overline{b}$ where $f$ is a definition, try $\overline{a\overset{?}{\equiv}b}$ backtracking if it fails. Otherwise unfold $s$ and $t$ and continue.
  \item Try $\eta$-expansion, replacing $s\overset{?}{\equiv} \lambda \overline{x}.\;t$ with $s\;\overline{x}\overset{?}{\equiv} t$.
  \item Try structure $\eta$, replacing $s\overset{?}{\equiv} (\overline{t_i})$ with $\overline{s.i\overset{?}{\equiv}t_i}$
  \item Unfold string literals.
  \item Try unit $\eta$, proving $s\equiv t$ if $s:\mathsf{Unit}$.
  \item Otherwise, fail. (Because of various incompletenesses in this algorithm, this doesn't actually imply $s\not\equiv t$.)
\end{itemize}

\section{Verifying the algorithm}\label{sec:verify}

We would now like to connect the results of \autoref{sec:base-theory} and \autoref{sec:typeck} by showing that the monadic typechecker functions satisfy their respective specifications. The typechecker monad \href{https://github.com/digama0/lean4lean/blob/cpp2026/Lean4Lean/TypeChecker.lean#L32}{\code{M}} is a standard reader/state/exception monad stack:
\begin{lstlisting}
abbrev $\mbox{\href{https://github.com/digama0/lean4lean/blob/cpp2026/Lean4Lean/TypeChecker.lean\#L32}{M}}$ :=
  ReaderT $\mbox{\href{https://github.com/digama0/lean4lean/blob/cpp2026/Lean4Lean/TypeChecker.lean\#L24}{Context}}$ <| StateT $\mbox{\href{https://github.com/digama0/lean4lean/blob/cpp2026/Lean4Lean/TypeChecker.lean\#L15}{State}}$ <| Except $\mbox{\href{https://github.com/leanprover/lean4/blob/v4.20.1/src/Lean/Environment.lean\#L251}{Exception}}$
\end{lstlisting}
Because we give the semantics of \href{https://github.com/leanprover/lean4/blob/v4.20.1/src/Lean/Expr.lean#L301}{\code{Expr}} using \href{https://github.com/digama0/lean4lean/blob/cpp2026/Lean4Lean/Theory/VExpr.lean#L6}{\code{VExpr}}, and this translation is only unique up to def.eq., we will maintain a shadow state \href{https://github.com/digama0/lean4lean/blob/cpp2026/Lean4Lean/Verify/TypeChecker/Theorems.lean#L111}{\code{VContext}} and \href{https://github.com/digama0/lean4lean/blob/cpp2026/Lean4Lean/Verify/TypeChecker/Theorems.lean#L143}{\code{VState}} which have \lstinline|Expr × VExpr| pairs everywhere so that we can project out both the original \href{https://github.com/digama0/lean4lean/blob/cpp2026/Lean4Lean/TypeChecker.lean#L24}{\code{Context}} as well as the corresponding \href{https://github.com/digama0/lean4lean/blob/cpp2026/Lean4Lean/Theory/VExpr.lean#L6}{\code{VExpr}} for each well-formed expression.

We use a lightweight monadic program logic based on Hoare logic, but ``one-sided'': for the precondition, we pass the exact context and state (and additional hypotheses become regular preconditions on the theorem), and for the postcondition we have a predicate over possible next state and return values.

\begin{lstlisting}
def $\mbox{\href{https://github.com/digama0/lean4lean/blob/cpp2026/Lean4Lean/Verify/TypeChecker/Theorems.lean\#L171}{M.WF}}$ (c : VContext) (vs : VState) (x : M α)
    (Q : α → VState → Prop) : Prop :=
  vs.WF c → ∀ a s',
  x c.toContext vs.toState = .ok (a, s') → ∃ vs',
  vs'.toState = s' ∧ vs ≤ vs' ∧ vs'.WF c ∧ Q a vs'
\end{lstlisting}
In words, this says that assuming $(c,vs)$ are \href{https://github.com/digama0/lean4lean/blob/cpp2026/Lean4Lean/Verify/TypeChecker/Theorems.lean#L145}{well-formed}, and denoting the projection to the machine state/context with $\widehat{\cdot}$, if running $x$ on $(\widehat{c},\widehat{vs})$ evaluates to $a$ in state $s'$ then there is some $vs'$ with $\widehat{vs'}=s'$ such that \href{https://github.com/digama0/lean4lean/blob/cpp2026/Lean4Lean/Verify/TypeChecker/Theorems.lean#L153}{$vs'$ is in the future of $vs$}, $(c,vs')$ is \href{https://github.com/digama0/lean4lean/blob/cpp2026/Lean4Lean/Verify/TypeChecker/Theorems.lean#L145}{well-formed}, and $Q(a,vs')$ holds. The future-of relation, denoted \href{https://github.com/digama0/lean4lean/blob/cpp2026/Lean4Lean/Verify/TypeChecker/Theorems.lean#L153}{$vs\le vs'$}, is a partial order on states which represent valid extension sequences. This is used mainly to handle the global counter used for creating fresh \lstinline|fvar|s. The well-formedness invariant ensures that the current counter value is greater than any \lstinline|fvar| currently in the state or context, and the \href{https://github.com/digama0/lean4lean/blob/cpp2026/Lean4Lean/Verify/TypeChecker/Theorems.lean#L153}{$vs\le vs'$} relation says that the global counter only ever goes up, so the state/context remains well-formed.

We can use this to define the specification of the \href{https://github.com/digama0/lean4lean/blob/cpp2026/Lean4Lean/TypeChecker.lean#L53}{\code{Methods}}:
\begin{lstlisting}
structure $\mbox{\href{https://github.com/digama0/lean4lean/blob/cpp2026/Lean4Lean/Verify/TypeChecker/Theorems.lean\#L203}{Methods.WF}}$ (m : Methods) where
  isDefEqCore : c.TrExpr e₁ e₁' → c.TrExpr e₂ e₂' →
    (m.isDefEqCore e₁ e₂).WF c s fun b _ =>
      b → c.IsDefEqU e₁' e₂'
  whnfCore : c.TrExpr e e' →
    (m.whnfCore e cr cp).WF c s
      fun e₁ _ => c.TrExpr e₁ e'
  whnf : c.TrExpr e e' →
    (m.whnf e).WF c s fun e₁ _ => c.TrExpr e₁ e'
  inferType : e.FVarsIn s.ngen.Reserves →
    (inferOnly = true → ∃ e', c.TrExpr e e') →
    M.WF c s (m.inferType e inferOnly) fun ty _ =>
      ∃ e' ty', c.TrExpr e e' ∧ c.TrExpr ty ty' ∧
        c.HasType e' ty'
\end{lstlisting}
\begin{itemize}
  \item \href{https://github.com/digama0/lean4lean/blob/cpp2026/Lean4Lean/TypeChecker.lean#L665}{\code{isDefEqCore}} (and \href{https://github.com/digama0/lean4lean/blob/cpp2026/Lean4Lean/TypeChecker.lean#L145}{\code{isDefEq}}) ensures that if it is called on well-typed expressions $e_1\leadsto e_1'$ and $e_2\leadsto e_2'$, and it returns \textsf{true}, then $e_1'\equiv e_2'$.
  \item \href{https://github.com/digama0/lean4lean/blob/cpp2026/Lean4Lean/TypeChecker.lean#L435}{\code{whnf}} and \href{https://github.com/digama0/lean4lean/blob/cpp2026/Lean4Lean/TypeChecker.lean#L310}{\code{whnfCore}} ensure that if it is called on a well-formed expression $e$ and returns $e_1$ then $e\equiv e_1$. We can succinctly say this by assuming $e\leadsto e'$ and asserting $e_1\leadsto e'$ as well, since $\leadsto$ is closed under def.eq.
  \item \href{https://github.com/digama0/lean4lean/blob/cpp2026/Lean4Lean/TypeChecker.lean#L240}{\code{inferType}} does not assume that it has a well-formed input, but it does need to know that the global counter is fresh for any fvars in the input. If \lstinline|inferOnly| is true, it also assumes $e$ is well-formed. It asserts that $e\leadsto e'$ and $\tau\leadsto \tau'$ are well-formed, and $e':\tau'$.
\end{itemize}
We then define \href{https://github.com/digama0/lean4lean/blob/cpp2026/Lean4Lean/Verify/TypeChecker/Theorems.lean#L203}{\code{RecM.WF}} with a similar signature to \href{https://github.com/digama0/lean4lean/blob/cpp2026/Lean4Lean/Verify/TypeChecker/Theorems.lean#L171}{\code{M.WF}} to mean that it is valid when called with a well-formed \href{https://github.com/digama0/lean4lean/blob/cpp2026/Lean4Lean/TypeChecker.lean#L53}{\code{Methods}}.

The specification of \href{https://github.com/digama0/lean4lean/blob/cpp2026/Lean4Lean/TypeChecker.lean#L115}{\code{inferLambda}} from \autoref{sec:inferType} should now be unsurprising:
\begin{lstlisting}
theorem $\mbox{\href{https://github.com/digama0/lean4lean/blob/cpp2026/Lean4Lean/TypeChecker.lean\#L765-L769}{inferLambda.WF}}$
  (h1 : e.FVarsIn s.ngen.Reserves)
  (hinf : inferOnly = true → ∃ e', c.TrExpr e e') :
  ($\mbox{\href{https://github.com/digama0/lean4lean/blob/cpp2026/Lean4Lean/TypeChecker.lean\#L115}{inferLambda}}$ e inferOnly).WF c s fun ty _ =>
    ∃ e' ty', c.TrExpr e e' ∧ c.TrExpr ty ty' ∧
      c.HasType e' ty'
\end{lstlisting}
The theorem is proved compositionally over the code. At the loop, we need an inductive invariant which reveals some of the complications of having a sequence of fvars that have not yet been applied to some of the inputs:
\begin{lstlisting}
theorem $\mbox{\href{https://github.com/digama0/lean4lean/blob/cpp2026/Lean4Lean/TypeChecker.lean\#L695-L705}{inferLambda.loop.WF}}$
  {c : VContext} {e₀ : Expr} {m : MLCtx}
  [c.MLCWF m] {n} (hn : n ≤ m.length)
  (hdrop : m.dropN n hn = c.mlctx)
  (harr : arr.toList.reverse =
    (m.fvarRevList n hn).map .fvar)
  (he₀ : e₀ = m.mkLambda n hn ei)
  (hei : e.instantiateList
    ((m.fvarRevList n hn).map .fvar) = ei)
  (hr1 : e.FVarsIn s.ngen.Reserves)
  (hr2 : ∀ v ∈ m.vlctx.fvars, s.ngen.Reserves v)
  (hinf : inferOnly = true → ∃ e',
    (c.withMLC m).TrExpr ei e') :
  ($\mbox{\href{https://github.com/digama0/lean4lean/blob/cpp2026/Lean4Lean/TypeChecker.lean\#L116}{inferLambda.loop}}$ inferOnly arr e).WF
      (c.withMLC m) s fun ty _ =>
    ∃ e', c.TrExpr e₀ e' ∧ ∃ ty', c.TrExpr ty ty' ∧ c.HasType e' ty'
\end{lstlisting}
Here $m:\mbox{\href{https://github.com/digama0/lean4lean/blob/cpp2026/Lean4Lean/TypeChecker.lean\#L47}{\code{MLCtx}}}$ is the aforementioned local context consisting of \lstinline|Expr × VExpr| pairs, which we use to locally change the local context while leaving the rest of the \href{https://github.com/digama0/lean4lean/blob/cpp2026/Lean4Lean/Verify/TypeChecker/Theorems.lean#L111}{\code{VContext}} alone.

\subsection{Bit tricks and soundness bugs}
Currently, only a few functions from the typechecker have been verified, like \href{https://github.com/digama0/lean4lean/blob/cpp2026/Lean4Lean/TypeChecker.lean#L115}{\code{inferLambda}}. However, in the course of performing the verification we have already run into some interesting issues, and at least one soundness bug (which has since been fixed in response to our report).

Recall that \href{https://github.com/digama0/lean4lean/blob/cpp2026/Lean4Lean/TypeChecker.lean#L115}{\code{inferLambda}} calls \href{https://github.com/digama0/lean4lean/blob/cpp2026/Lean4Lean/Instantiate.lean#L8}{\code{cheapBetaReduce}}, so we need to prove that \href{https://github.com/digama0/lean4lean/blob/cpp2026/Lean4Lean/Instantiate.lean#L8}{\code{cheapBetaReduce}} (which morally performs only $\beta$-reduction) preserves def.eq. The \href{https://github.com/digama0/lean4lean/blob/cpp2026/Lean4Lean/Instantiate.lean#L8}{\code{cheapBetaReduce}} function in turn uses \href{https://github.com/leanprover/lean4/blob/v4.20.1/src/Lean/Expr.lean#L1267}{\code{hasLooseBVars}} to determine whether the function is a constant function (since $\lambda \overline{x}.\;c$ is a constant function if $c$ has no bvars after stripping all the binders without instantiating them). \href{https://github.com/leanprover/lean4/blob/v4.20.1/src/Lean/Expr.lean#L1267}{\code{hasLooseBVars}} $e$ is defined as \lstinline|$\mbox{\href{https://github.com/leanprover/lean4/blob/v4.20.1/src/Lean/Expr.lean\#L590}{looseBVarRange}}$ $e$ > 0|, so we are reduced to verifying \href{https://github.com/leanprover/lean4/blob/v4.20.1/src/Lean/Expr.lean#L590}{\code{looseBVarRange}}, which returns the largest unbound bvar.

One would expect this to be defined recursively over \href{https://github.com/leanprover/lean4/blob/v4.20.1/src/Lean/Expr.lean#L301}{\code{Expr}}:
\begin{align*}
\mathsf{LBR}(\mathsf{bvar}\;i)&=i+1,\\
\mathsf{LBR}(\mathsf{lam}\;\tau\;e)&=\max(\mathsf{LBR}(\tau),\mathsf{LBR}(e)-1), \quad \mbox{etc.}
\end{align*}
However, it is instead implemented by masking off 20 bits from a 64-bit data field \href{https://github.com/leanprover/lean4/blob/v4.20.1/src/Lean/Expr.lean#L470}{$e.\mathsf{data}$} which is stored with every expression. So this naturally leads to the question: what happens on overflow? The calculation of \href{https://github.com/leanprover/lean4/blob/v4.20.1/src/Lean/Expr.lean#L470}{$e.\mathsf{data}$} had an overflow check, but it was expressed in lean using the \href{https://github.com/leanprover/lean4/blob/v4.20.1/src/Init/Prelude.lean#L2860}{\code{panic!}} function, which signals an error on \textsf{stdout} but then continues execution with the default value of the type, which is \href{https://github.com/leanprover/lean4/blob/v4.20.1/src/Init/Prelude.lean#L2354}{$0$ for \code{UInt64}}, which also happens to be the worst possible answer in this case, since it effectively turns on all of the optimizations, including the one in question about \href{https://github.com/leanprover/lean4/blob/v4.20.1/src/Lean/Expr.lean#L1267}{\code{hasLooseBVars}}. This is a nice example of a bug that was found entirely theoretically, by failing to prove a theorem and working backwards to an actual exploitable soundness bug.

Unfortunately, the fix does not address all of the problems we identified. The way it was addressed was to make the \href{https://github.com/leanprover/lean4/blob/v4.20.1/src/Lean/Expr.lean#L470}{\code{Expr.data}} function opaque, which at most masks the issue; and also by moving the panic to C++ code, where it can ``panic harder'' and crash the program, so that the strange consequences of returning $0$ are not observed. Using opaques to replace a low level definition with a model is a reasonable option, but currently we are still relying on an unsound assumption (\href{https://github.com/digama0/lean4lean/blob/cpp2026/Lean4Lean/Verify/Axioms.lean#L164-L165}{\code{looseBVarRange\_eq}}), because the definition of \href{https://github.com/leanprover/lean4/blob/v4.20.1/src/Lean/Expr.lean#L590}{\code{looseBVarRange}} is still visible and so one can prove that e.g. \lstinline|$\mbox{\href{https://github.com/leanprover/lean4/blob/v4.20.1/src/Lean/Expr.lean\#L590}{looseBVarRange}}$ $e$ < 2^20| unconditionally, which means that \href{https://github.com/leanprover/lean4/blob/v4.20.1/src/Lean/Expr.lean#L590}{\code{looseBVarRange}} cannot be the ``correct'' definition, which means the theorem isn't provable even though all counterexamples are unreachable (but Lean's logic doesn't know that).

\section{The global environment}\label{sec:env}

The environment is the global structure that ties together individual declarations. It has made some appearances in the previous sections already, because the environment is needed to typecheck constants as well as definitional extensions, given in the \textsc{t-const} and \textsc{t-extra} rules.

The Lean type for this, \href{https://github.com/leanprover/lean4/blob/v4.20.1/src/Lean/Environment.lean#L198}{\code{Environment}}, is complex and contains many details irrelevant to the kernel, but luckily we only really care about a few operations on it:

\begin{itemize}
  \item $\mbox{\href{https://github.com/digama0/lean4lean/blob/cpp2026/Lean4Lean/Environment/Basic.lean\#L64}{\code{empty}}} : \mbox{\lstinline|Environment|}$\\ -- constructs an environment with no constants
  \item $\mbox{\href{https://github.com/leanprover/lean4/blob/v4.20.1/src/Lean/Environment.lean\#L302}{\code{add}}} : \mbox{\lstinline|Environment|}\to\mbox{\lstinline|ConstantInfo|}\to\mbox{\lstinline|Environment|}$\\ -- adds a \lstinline|ConstantInfo| declaration to the environment
  \item $\mbox{\href{https://github.com/leanprover/lean4/blob/v4.20.1/src/Lean/Environment.lean\#L811}{\code{find?}}} : \mbox{\lstinline|Environment|} \to \mbox{\lstinline|Name|} \to \mbox{\lstinline|Option ConstantInfo|}$\\ -- retrieves a declaration by name
\end{itemize}
In effect, the environment is just a fancy hashmap which indexes \href{https://github.com/leanprover/lean4/blob/v4.20.1/src/Lean/Declaration.lean#L417}{\code{ConstantInfo}} declarations by their names.

The corresponding theory type is called \href{https://github.com/digama0/lean4lean/blob/cpp2026/Lean4Lean/Theory/VEnv.lean#L15}{\code{VEnv}}, with the definition:
\begin{lstlisting}
structure $\mbox{\href{https://github.com/digama0/lean4lean/blob/cpp2026/Lean4Lean/Theory/VEnv.lean\#L5}{VConstant}}$ where
  (uvars : Nat) (type : VExpr)
structure $\mbox{\href{https://github.com/digama0/lean4lean/blob/cpp2026/Lean4Lean/Theory/VEnv.lean\#L9}{VDefEq}}$ where
  (uvars : Nat) (lhs rhs type : VExpr)
structure $\mbox{\href{https://github.com/digama0/lean4lean/blob/cpp2026/Lean4Lean/Theory/VEnv.lean\#L15}{VEnv}}$ where
  constants : Name → Option (Option VConstant)
  defeqs : VDefEq → Prop
\end{lstlisting}
This is a very simple type, essentially exactly what is needed to satisfy the requirements of the typing judgment.
\begin{itemize}
  \item \lstinline|constants| maps a name to a constant, if defined, where a constant is given by its universe arity and its type. It returns \lstinline|none| if a constant by that name does not exist, and \lstinline|some none| if the name has been ``blocked'', meaning that there is no constant there but we still want to make it illegal to make a definition with that name.\footnote{This is how we model \lstinline|unsafe| declarations when typechecking safe declarations: they ``don't exist'' for most purposes, but it is nevertheless not allowed to shadow them with another safe declaration.} This corresponds to the $\bar u.(c_{\bar u}:\alpha)$ declaration appearing in \textsc{t-const}.
  \item \lstinline|defeqs| is a set of (anonymous) \href{https://github.com/digama0/lean4lean/blob/cpp2026/Lean4Lean/Theory/VEnv.lean#L9}{\code{VDefEq}}s, containing two expressions to be made definitionally equal and their type. This corresponds to the $\bar u.(e\equiv e':\alpha)$ declaration appearing in \textsc{t-extra}.
\end{itemize}

In some sense this doesn't actually answer many interesting questions about where the constants come from or what definitional equalities are permitted. For that, we need \href{https://github.com/digama0/lean4lean/blob/cpp2026/Lean4Lean/Theory/VDecl.lean#L22}{\code{VDecl}}, the type of records that can be used to update the environment. The declarations are:
\begin{itemize}
  \item \lstinline|block $n$| blocks constant $n$, setting \lstinline|env.constants $n$ = some none| and preventing a redeclaration as mentioned above.
  \item \lstinline|axiom { name := $c$, uvars := $n$, type := $\alpha$ }|\\ adds $c:\alpha$ as an axiom to the constant map.
  \item \lstinline|opaque { $c$, $n$, $\alpha$, value := $e$ }| checks that $e:\alpha$, then adds $c:\alpha$ to the constant map.
  \item \lstinline|def { $c$, $n$, $\alpha$, $e$ }| does the same as \lstinline|opaque| but also adds $c\equiv e:\alpha$ to the defeq set.
  \item \lstinline|example { uvars := $n$, type := $\alpha$, value := $e$ }|\\ checks that $e:\alpha$ and leaves the environment as-is.
  \item \lstinline|induct ($d$ : VInductDecl)| adds all the constants and defeqs from an inductive declaration (see \autoref{sec:inductive}).
  \item \lstinline|quot| adds the quotient axioms and definitional equality. This is a type operator $\alpha/R$ where $R:\alpha\to\alpha\to\mathtt{Prop}$ with $\mbox{mk}:\alpha\to\alpha/R$ and $\mbox{lift}:(f:\alpha\to\beta)\to(\forall x y.\;R\;x\;y\to f\;x=f\;y)\to (\alpha/R\to\beta)$, with the property that $(\mbox{lift}\;f\;h\;(\mbox{mk}\;a)\equiv f\;a:\beta)$. (One can construct every part of this in Lean except for the definitional equality.)
\end{itemize}

This enumeration is fairly similar to the \href{https://github.com/leanprover/lean4/blob/v4.20.1/src/Lean/Declaration.lean#L173}{\code{Declaration}} type which is the actual front end to the kernel, but it lacks \lstinline|unsafe| declarations and \lstinline|mutualDefnDecl| (which, despite the name, is not the way mutual definitions are sent to the kernel, but rather represents unsafe declarations with unchecked self-referential definitions). For the most part we do not attempt to model \lstinline|unsafe| declarations because these are typechecked with certain checks disabled, and this makes them no longer follow the theory, but this is deliberate. There isn't much we can say about such definitions since they can violate type safety and cause undefined behavior, and in any case they don't play a role in checking theorems.

\section{Inductive types}\label{sec:inductive}

Lean's implementation of inductive types is based on Dybjer \cite{dybjer}. It notably differs from Rocq in that rather than having primitive \textsf{fix} and \textsf{match} expressions, we have constants called ``recursors'' generated from the inductive specification. This choice significantly simplifies the \href{https://github.com/digama0/lean4lean/blob/cpp2026/Lean4Lean/Theory/VExpr.lean#L6}{\code{VExpr}} type, which as we have seen contains no special support for anything inductive-related except for the generic def.eq. hook \textsc{t-extra}. It also means we do not need a complex guard checker, something which MetaRocq currently axiomatizes \cite{coqcoqcorrect} and which has been the source of soundness bugs in the past. Instead, the complexity is pushed to the generation of the recursor for the inductive type.\footnote{It would be premature to draw conclusions about the relative complexity of these two approaches though, as many of the important theorems are unfinished.}

For single inductives, the algorithm is described in full in \cite{leantt}. In Lean 3, nested and mutual inductives were simulated using single inductives so the kernel only had to deal with the single inductive case, but for performance reasons Lean 4 moved the checking to the kernel. For a simple example of a mutual inductive, we can define \lstinline|Even| and \lstinline|Odd| by mutual induction like so:
\begin{lstlisting}
mutual
inductive Even : Nat → Prop where
  | zero : Even 0
  | succ : Odd n → Even (n+1)

inductive Odd : Nat → Prop where
  | succ' : Even n → Odd (n+1)
end
\end{lstlisting}
which generates the definitions:
\begin{lstlisting}
Even : Nat → Type
Even.zero : Even 0
Even.succ : ∀ {n : Nat}, Odd n → Even (n + 1)
Odd : Nat → Type
Odd.succ' : ∀ {n : Nat}, Even n → Odd (n + 1)
Even.rec.{u} :
  {motive_1 : (a : Nat) → Even a → Sort u} →
  {motive_2 : (a : Nat) → Odd a → Sort u} →
  motive_1 0 Even.zero →
  ({n : Nat} → (a : Odd n) → motive_2 n a → motive_1 (n + 1) (Even.succ a)) →
  ({n : Nat} → (a : Even n) → motive_1 n a → motive_2 (n + 1) (Odd.succ' a)) →
  {a : Nat} → (t : Even a) → motive_1 a t
Odd.rec.{u} :
  ... → -- same as Even.rec
  {a : Nat} → (t : Odd a) → motive_2 a t
\end{lstlisting}
and the definitional equalities:
\begin{lstlisting}
@Even.rec m1 m2 Fz Fs Fs' 0 Even.zero ≡ Fz
@Even.rec m1 m2 Fz Fs Fs' (n+1) (Even.succ e) ≡
  @Fs n e (@Odd.rec m1 m2 Fz Fs Fs' n e)
@Odd.rec m1 m2 Fz Fs Fs' (n+1) (Odd.succ' e) ≡
  @Fs' n e (@Even.rec m1 m2 Fz Fs Fs' n e)
\end{lstlisting}

For nested inductives, the checking procedure effectively amounts to running the simulation procedure from Lean 3 to get a mutual inductive type and performing well-formedness checks there, then removing all references to the unfolded type in the recursor.

Currently, \lfl contains a complete implementation of inductive types, including nested and mutual inductives and all the typechecking effects of this, but not much work has been done on the theoretical side, defining what an inductive specification should generate.

\subsection{Eta for structures}\label{sec:eta}

Another new kernel feature in Lean 4 is primitive projections, and $\eta$ for structures, which says that if $s$ is a \lstinline|structure| (which is to say, a one-constructor non-recursive \lstinline|inductive| type), then $s\equiv\mathsf{mk}(s.1,\dots,s.n)$ where $s.i$ is the $i$'th projection (\href{https://github.com/leanprover/lean4/blob/v4.20.1/src/Lean/Expr.lean#L452-L467}{\code{Expr.proj}}). This was already true as a propositional equality proven by induction on $s$, but it is now a definitional equality and the kernel has to unfold structures when needed to support this. $\eta$ for unit in particular is known to be problematic \cite{eta-unit,lennon-bertrand2022}, but Lean's kernel does not implement a complete decision procedure anyway. We believe it poses no issue for soundness and the main results should still hold even with this extension.

\section{Results}\label{sec:results}

\lfl contains a command-line frontend, which can be used to validate already-compiled Lean projects. It operates on \texttt{.olean} files, which are essentially serialized (unordered) lists of \href{https://github.com/leanprover/lean4/blob/v4.20.1/src/Lean/Declaration.lean#L417}{\code{ConstantInfo}} objects, and the frontend topologically sorts these definitions and passes them to the \lfl kernel, which is a function
\begin{lstlisting}
def $\mbox{\href{https://github.com/digama0/lean4lean/blob/cpp2026/Lean4Lean/Environment.lean\#L98-L99}{addDecl}}$ (env : Environment)
    (decl : Declaration) (check := true) :
  Except KernelException Environment
\end{lstlisting}
that is a drop-in replacement for \href{https://github.com/leanprover/lean4/blob/v4.20.1/src/Lean/Environment.lean#L288-L289}{\code{Environment.addDeclCore}}, an opaque Lean function implemented by FFI (for\-eign-func\-tion interface) to the corresponding C++ kernel function.

This is the same interface (and to some extent, the same code) as \texttt{lean4checker},\footnote{\url{https://github.com/leanprover/lean4checker}} which does the same but calls Lean's \href{https://github.com/leanprover/lean4/blob/v4.20.1/src/Lean/AddDecl.lean#L14-L15}{\code{addDecl}} instead, using the C++ kernel as an external verifier. It may seem odd to use the Lean kernel as an external verifier for itself, but it can catch cases where malicious (or confused) Lean code uses of the very powerful metaprogramming framework to simply bypass the kernel and add constants to the environment without typechecking them.
For our purposes, it acts as a very handy benchmark for comparison, since we are performing the same operations but with a different kernel. Running \verb|lake env lean4lean| \verb|--fresh Mathlib| in \texttt{mathlib4} will check the \textsf{Mathlib} module and all of its dependencies.\footnote{This is running in single-threaded mode; it can also check modules individually and in parallel, assuming the correctness of imports, but it runs into memory usage limitations so is harder to benchmark reliably.}

\begin{figure}[tb]
  \centering
  \begin{tabular}{l|c|c|c}
    & \texttt{lean4checker} & \texttt{lean4lean} & ratio \\ \hline
    \textsf{Lean} & 37.01 s & 44.61 s & 1.21 \\
    \textsf{Batteries} & 32.49 s & 45.74 s & 1.40 \\
    \makecell[lc]{\textsf{Mathlib} \\ (+ \textsf{B.} + \textsf{Lean})} & 44.54 min & 58.79 min & 1.32 \\
  \end{tabular}
  \caption{Comparison of the original C++ implementation of the Lean kernel (\texttt{lean4checker}) with \texttt{lean4lean} (Lean) on major Lean packages. Tests were performed on a 12 core 12th Gen Intel i7-1255U @ 2.1 GHz, single-threaded, on rev. 526c94c of \texttt{mathlib4}.}
  \label{fig:comparison}
\end{figure}

The results are summarized in \autoref{fig:comparison}. The new Lean implementation is around 30\% slower than the C++ implementation, which we attribute mainly to shortcomings in the Lean compiler and data representation compared to C++. Nevertheless, it is within an order of magnitude and practically usable on large libraries, which we consider a strong sign.
We do not expect this to get any faster without Lean compiler improvements, but we also intend to keep the speed from regressing as we add more verification. Because the verification is external to the executable itself, this is possible in theory, but more tooling may be needed to make it possible to verify the implementation as-is in practice.

\section{Related work}\label{sec:related-work}

Self-verification of ITP systems has been done in Rocq \cite{coqincoq}, HOL Light \cite{harrison-holhol}, Isabelle \cite{nipkow21}, and Metamath Zero \cite{mm0-paper}, and to some extent self-verification also overlaps with bootstrapping theorem provers such as Milawa \cite{davis09} and Candle \cite{candle,Kumar14}. But the most similar work is unquestionably MetaRocq \cite{sozeau20,coqcoqcorrect,sozeau23}, which is a project to develop a verified typechecker for Rocq, in Rocq. It is significantly more complete than the previous effort \cite{coqincoq}, but is not yet capable of verifying the Rocq standard library due to lack of eta-conversion and template polymorphism; eta-conversion is work in progress \cite{lennon-bertrand2022}, and universe polymorphism being reorganised in Rocq's kernel \cite{poiret2025}. With \lfl we took the approach of quickly getting to feature parity with the real Lean kernel and proving correctness later, so we ended up with a feature-complete typechecker but are still lacking in some theorems. Our work also overlaps to some extent with Template-Rocq \cite{anand18,boulierphd,winterhalterphd}, a certified metaprogramming framework for Rocq developed in the MetaRocq repo.

As was mentioned, MetaRocq also has to deal with many of the same issues proving properties of the typing judgment, and one may hope for proofs from MetaRocq to be useful in \lfl and vice-versa. Unfortunately, there are some important details that differ:
\begin{itemize}
  \item In MetaRocq, the conversion relation is a partial order rather than an equivalence relation, because of universe cumulativity, so \sref{Conjecture}{conj:IsDefEq.uniq} doesn't hold as stated. Nevertheless, there are theorems regarding ``principal types'' which can be used to play the same role.
  \item Unlike in \lfl, the conversion relation is \emph{untyped}: there is no mutual induction between typing and definitional equality, which massively simplifies matters. In Lean this is unfortunately not an option because of the \textsc{t-proof-irrel} rule, which is expressed in a different way in Rocq.
  \item Lean's typechecker is known to not be complete, owing to some of the undecidability results in \cite{leantt}, but MetaRocq \cite{sozeau23,lennonbertrandphd} has a proof of completeness.
\end{itemize}

\section{Conclusion}\label{sec:conclusion}

We have presented \lfl, a reimplementation of the Lean 4 typechecker written in Lean itself and closely following the kernel. Despite its simplicity, \lfl can validate large-scale developments such as \textsf{mathlib} with competitive performance.

More importantly, \lfl serves as a basis for formal reasoning about Lean’s metatheory. By internalizing the type system, we can verify properties of the checker and use these results to increase confidence in the production kernel. Already, this process has revealed subtle issues and motivated fixes upstream.

Looking ahead, we plan to formalize more of \cite{leantt}, including in particular the proof of relative consistency given a large cardinal assumption. By continuing the verification effort, we hope to discover more potential bugs in the production kernel. In the long term, we hope that this work will lead to a fully verified, self-hosted Lean kernel.

\begin{acks}
I would like to thank Jeremy Avigad for their support and encouragement, Yannick Forster and the anonymous reviewers for reading early drafts of this work, and Matthieu Sozeau and the rest of the MetaRocq team for collaborating and taking interest in this work. Work supported by the Hoskinson Center for Formal Mathematics at Carnegie Mellon.
\end{acks}


\bibliographystyle{ACM-Reference-Format}
\bibliography{references}

\end{document}